\documentclass[reprint,longbibliography,aps,prl,table,superscriptaddress]{revtex4-2}
\usepackage[english]{babel}
\usepackage[T1]{fontenc}
\usepackage{amsmath,amsthm,amssymb,bbm,amsfonts,graphicx,framed,mathtools,dsfont}
\usepackage{url}
\usepackage{tikz}
\usetikzlibrary{decorations.pathreplacing,decorations.pathmorphing}
\usepackage{enumerate}
\usepackage{mathrsfs}
\usepackage{relsize}
\usepackage[colorlinks=true, allcolors=blue]{hyperref}
\usepackage[utf8]{inputenc}
\usepackage{listings}
\usepackage{stmaryrd}
\usepackage{graphicx}
\usepackage{longtable}

\lstdefinestyle{mystyle}{
  backgroundcolor=\color{backcolour},
  commentstyle=\color{codegreen},
  keywordstyle=\color{magenta},
  numberstyle=\tiny\color{codegray},
  stringstyle=\color{codepurple},
  basicstyle=\ttfamily\footnotesize,
  breakatwhitespace=false,
  breaklines=true,
  captionpos=b,
  keepspaces=true,
  numbers=left,
  numbersep=5pt,
  showspaces=false,
  showstringspaces=false,
  showtabs=false,
  tabsize=2
}
\newcommand{\ka}{\kappa}

\newcommand{\la}{\lambda}

\newcommand{\si}{{\sigma}}

\newcommand{\eq}[1]{Eq.~(\ref{eq:#1})}
\newcommand{\fig}[1]{Fig.~\ref{fig:#1}}

\newcommand{\mc}[1]{\mathcal{#1}}
\newcommand{\id}{\mathbbm{1}}
\newcommand{\tr}{\mathrm{tr}}

\newcommand{\ket}[1]{\left\lvert #1 \right\rangle} 
\newcommand{\bra}[1]{\left\langle #1 \right\rvert} 

\let\perptmp\perp
\renewcommand{\perp}{{\! \mathsmaller{\perptmp}}}

\newcommand{\mrm}{\mathrm}

\newcommand{\nodagger}{{\phantom{\dagger}}}

\newcommand{\dd}{\mathrm{d}}

\newtheorem*{proposition}{Proposition}

\theoremstyle{definition}

\renewcommand{\sp}[2]{\langle #1,#2 \rangle}

\makeatletter
\def\@bibdataout@aps{%
  \immediate\write\@bibdataout{%
    @CONTROL{%
      apsrev41Control%
      \longbibliography@sw{%
        ,author="08",editor="1",pages="1",title="0",year="1"%
      }{%
        ,author="08",editor="1",pages="1",title="",year="1"%
      }%
    }%
  }%
  \if@filesw \immediate \write \@auxout {\string \citation {apsrev41Control}}\fi
}
\makeatother

\definecolor{blue}{rgb}{0.12156862745098039, 0.4666666666666667, 0.7058823529411765}
\definecolor{orange}{rgb}{1.0, 0.4980392156862745, 0.054901960784313725}
\definecolor{green}{rgb}{0.17254901960784313, 0.6274509803921569, 0.17254901960784313}
\definecolor{red}{rgb}{0.8392156862745098, 0.15294117647058825, 0.1568627450980392}
\definecolor{purple}{rgb}{0.5803921568627451, 0.403921568627451, 0.7411764705882353}
\definecolor{brown}{rgb}{0.5490196078431373, 0.33725490196078434, 0.29411764705882354}
\definecolor{pink}{rgb}{0.8901960784313725, 0.4666666666666667, 0.7607843137254902}
\definecolor{pygray}{rgb}{0.4980392156862745, 0.4980392156862745, 0.4980392156862745}
\definecolor{olive}{rgb}{0.7372549019607844, 0.7411764705882353, 0.13333333333333333}
\definecolor{9}{rgb}{0.09019607843137255, 0.7450980392156863, 0.8117647058823529}

\definecolor{lightblue}{rgb}{0.9,.95,1}

\newcommand{\diag}{\mathrm{diag}}

\usepackage{calc}

\usepackage{pgffor}

\begin{document}

\title{Non-Markovianity induced by Pauli-twirling}
 
\author{Joris Kattemölle}
\email{j.kattemoelle at fz-juelich dot de}
\thanks{This author contributed equally to this work.}
\affiliation{Department of Physics and IQST, University of Konstanz, 78457 Konstanz, Germany}
\affiliation{Institute for Theoretical Nanoelectronics (PGI-2),
Forschungszentrum Jülich, 52428 Jülich, Germany}
\affiliation{Institute for Quantum Information, RWTH Aachen University, 52056 Aachen, Germany}
\author{Balázs Gulácsi}
\email{balazs.gulacsi@uni-konstanz.de}
\thanks{This author contributed equally to this work.}
\affiliation{Department of Physics and IQST, University of Konstanz, 78457 Konstanz, Germany}
\author{Guido Burkard}
\email{guido.burkard@uni-konstanz.de}
\affiliation{Department of Physics and IQST, University of Konstanz, 78457 Konstanz, Germany}

\date{\today}

\begin{abstract} 
Noise forms a central obstacle to effective quantum information processing. Recent experimental advances have enabled the tailoring of noise properties through Pauli twirling, transforming arbitrary noise channels into Pauli channels. This underpins theoretical descriptions of fault-tolerant quantum computation and forms an essential tool in noise characterization and error mitigation. Pauli-Lindblad channels have been introduced to aptly parameterize quasi-local Pauli errors across a quantum register, excluding  
negative Pauli-Lindblad parameters relying on the Markovianity of the underlying noise processes.
We point out that caution is required when parameterizing channels as Pauli-Lindblad channels with nonnegative parameters. 
For this, we study the effects of Pauli twirling on Markovianity. We use the notion of Markovianity of a channel (rather than that of an entire semigroup) and prove a general Pauli channel is non-Markovian if and only if at least one of its Pauli-Lindblad parameters is negative. Using this, we show that Markovian quantum channels often become non-Markovian after Pauli twirling.  
The Pauli-twirling induced non-Markovianity necessitates the use of negative Pauli-Lindblad parameters for a correct noise description in experimentally realistic scenarios. 
An important example is the implementation of the $\sqrt{X}$-gate under standard Markovian noise.
As such, our results have direct implications for quantum error mitigation protocols that rely on accurate noise characterization.
\end{abstract}

\maketitle

\section{\label{sec:level1}Introduction}

Over the past decades, quantum computing has evolved from a theoretical curiosity~\cite{Feynman1982,Deutsch1985} to a practical reality. Despite impressive experimental progress~\cite{Cirac1995,Monroe1995,Chuang1998,Nakamura1999,Koch2007,Barends2014,QEC2023,QEC2025}, current devices remain far from ideal, with noise threatening to wash away the very quantum effects that give them an advantage over their classical counterparts. Among others, this noise arises from a combination of imperfect quantum gate implementations~\cite{Magesan2011,Magesan2012}, crosstalk~\cite{Mundada2019,Heinz2021,ParradoRodriguez2021}, and uncontrolled interactions with the surrounding environment~\cite{Chirolli2008,Paladino2014}. The central obstacle to utility-scale quantum computation is therefore the understanding and subsequent reduction of noise processes and the errors they induce.

Meeting this challenge requires accurate error models. Mathematically, the effect of quantum noise can be described by quantum channels; completely positive and trace-preserving linear maps on density matrices. One class of such channels are the Pauli channels, where a tensor product of Pauli operators $P_a$
(known as a Pauli word or a Pauli string) is applied to the system with probability $p_a$. Pauli channels are particularly appealing due to the relative ease with which they can be characterized, analyzed and implemented for simulations on a classical computer~\cite{flammia2020efficient}. The empirical success of Pauli channels as effective noise models is largely enabled by Pauli twirling, also known as randomized compiling~\cite{Wallman2016rc}. This technique allows an arbitrary noise channel to be ``twirled'' into an effective Pauli channel and has become a standard method~\cite{Hashim2021,Ware2021,Ville2022,Gu2023,Jain2023}. A quantum algorithm is implemented by a quantum circuit that is compiled into a sequence of elementary quantum gates. To perform randomized compiling, Pauli words are injected into the circuit before and after certain gates in a way that preserves the overall logical operation of the circuit. Repeating this process with different random choices produces a set of randomized but logically equivalent circuits. Averaging the outcomes over these produces an effective noise channel in Pauli form. 

Although representing noise as a Pauli channel has its theoretical appeal and experimental practicality, it does not eliminate all challenges. As the number of qubits $n$ increases, the parameter count required to describe a general Pauli channel grows exponentially, rendering a full description intractable and experimentally inaccessible.
To address this scaling problem, it has recently been proposed to parameterize Pauli noise channels using a sparse set of Lindbladian generators~\cite{berg2023probabilistic}. These Pauli-Lindblad (PL) channels are written as
\begin{gather}\label{eq:PL}
    \mc E(\rho)=e^{\mc L}\rho,\quad\mc L (\rho)=\sum_{a\in\mathcal S}\la_a(P_a \rho P_a-\rho),
\end{gather}
where $\mc S$ is a subset of all $4^n$ possible Pauli words and on the grounds of physicality \cite{vdBevidence2023}, the PL parameters are assumed to be nonnegative, $\lambda_a\geq0$. The complexity of this model is determined by the size of the subset $\mathcal S$, which scales polynomially with $n$ under the assumption of a bounded noise correlation length. PL models have been studied theoretically \cite{Tax2022,berg2024techniques,Malekakhlagh2025} and experimentally~\cite{berg2023probabilistic,vdBevidence2023}, where the effectiveness of the model is evidenced by the implementation of error-mitigation techniques such as probabilistic error cancellation and zero-noise extrapolation based on these models.

Markovianity in quantum systems is usually seen as a property of a physical process, formalizing in one way or another that at any time, future states of the system only depend on the current state of the systems, and not previous ones. 
Using this notion, Markovianity can also be defined for a single quantum channel $\mc E$, if there exists a Markovian physical process such that after an elapsed time $t$, the effect of the process is the channel $\mc E$. 

In this work, we study the interplay of Pauli-twirling and Markovianity using the framework of open system dynamics. First, we show that any Pauli channel can be parameterized as a PL channel if we lift the restriction that the PL parameters are nonnegative. The (generalized) PL channels with at least one negative PL parameter can in fact be physical and turn out to exactly coincide with the non-Markovian Pauli channels. We apply this characterization of non-Markovianity and find that, perhaps surprisingly,  Pauli-twirling can break the Markovianity of a quantum channel. 

We present two examples where Markovianity-breaking arises. The first is a single-qubit protocol subject to Markovian dephasing and relaxation. The protocol is explicitly designed to be simple, yet showing strong Pauli-twirling-induced non-Markovianity. A second example is the noise channel of the standard implementation of the $\sqrt{X}$ gate under Markovian dephasing and relaxation. We estimate that these effects can be observed under realistic experimental settings, especially in dephasing-biased quantum hardware, and may in fact already be present in current implementations of the $\sqrt{X}$ gate. 

These insights are particularly important for quantum error mitigation protocols, for which accurate noise descriptions are essential~\cite{QEMreview2023}. Indeed, the possibility of negative PL model parameters, even if their magnitude is smaller than that of the positive ones, introduces a qualitative modification of the noise models, and their inclusion is therefore crucial for an accurate description and correct error-mitigated results. 

\subsection{Preliminaries}
\paragraph{\normalfont\small\textbf{Markovianity by divisibility.}}
There are numerous nonequivalent mathematical definitions of Markovianity that capture the concept of memoryless quantum evolution. One definition is based on the divisibility of families of quantum channels. Consider a two-parameter family of quantum channels $\mathscr E=\{\mc E_{t_3,t_1}\mid t_3\geq t_1 \geq 0\}$, where $\mc E_{t_3,t_1}$ describes the time evolution of a quantum system from $t_1$ to $t_3$. This family is called \emph{divisible} if for any intermediate time $t_2$, the total time evolution can be concatenated, i.e., if for all $t_3\geq t_2 \geq t_1 \geq 0$, we have  $\mc E_{t_3,t_1}=\mc E_{t_3,t_2}\circ\mc E_{t_2,t_1}$ \cite{rivas2014quantum}. We say that the two-parameter family $\mathscr E$ is \emph{Markovian by divisibility} if it satisfies the divisibility condition.

A fundamental result holds when the two-parameter family $\mathscr E$ is differentiable. 
The Gorini-Kossakowski-Sudarshan-Lindblad (GKSL) theorem~\cite{rivas2014quantum}  states that
an operator $\mc L_t$ 
is the generator of the divisible family of quantum channels if and only if it can be written in the form
\begin{equation}\label{eq:generator}
\frac{\dd \rho}{\dd t}=\mc L_t[\rho]=-i[H(t),\rho]+\mathcal D_t[\rho],
\end{equation}
with $H(t)$ Hermitian and with the dissipator
\begin{equation}\label{eq:dissipator}
\mathcal D_t[\rho]=\sum_{a,b} \Gamma_{ab}(t)\left[F_a \rho F_b^\dagger-\frac{1}{2}\{F_b^\dagger F_a,\rho\}\right],
\end{equation}
where $\Gamma_{ab}(t)\geq 0$ (positive semidefinite) for all $t\geq 0$. 

The division into the Hamiltonian part $-i[H(t),\rho]$ and a dissipative part $\mc D_t[\rho]$ is unique if the sums run over an orthogonal operator basis $F_a$ [$\tr(F_a^\dagger F_b)\propto\delta_{ab}$]. The Kossakowski matrix $\Gamma$ can be diagonalized as $\Gamma_{ab}(t)=\sum_ku_{ak}(t)\gamma_k(t)u_{bk}^*(t)$, yielding the familiar equivalent form for the dissipator
$\mathcal D_t[\rho]=\sum_{k} \gamma_k(t)\left[L_k(t) \rho L^\dagger_k(t)-\frac{1}{2}\{L_k^\dagger(t)L_k(t),\rho\}\right],$
with $\gamma_k(t)\geq 0$, and with jump operators $L_k(t)=\sum_a u_{ak}(t)F_a$. The operator $\mc L_t$, at a fixed $t$, is said to be in Lindblad form if the $\gamma_k(t)$ are nonnegative. 

The assumption of Markovianity by divisibility is ubiquitous in the modeling of gate-based quantum computing. At the circuit level, the total evolution of the qubits is discretized into time steps, each corresponding to a circuit layer. In general, the state of the system after each time step evolves according to a complex quantum stochastic process~\cite{Modi2021prx}. The assumption of Markovianity by divisibility allows us to approximate this dynamics as an independent sequence of quantum channels. At the gate level, the assumption of Markovianity by divisibility leads to a description in terms of the time-dependent GKSL equation [Eq.~\eqref{eq:generator}]. In this case, 
$H(t)$ represents the control Hamiltonian that implements the intended unitary operation, while the dissipator captures noise processes such as relaxation and dephasing.

\paragraph{\normalfont\small\textbf{Channel semigroup Markovianity.}} When the generator appearing in Eq.~\eqref{eq:generator} is time-independent, $\mc L_t=\mc L$, the resulting family of channels forms a one-parameter semigroup, $\mathscr E=\{\mc E_t \mid t\geq 0\}$, with $\mc E_{t_2}\mc \circ \mc E_{t_1}=\mc E_{t_2+t_1}$. In this case, each channel $\mc E_t$ can be written as $\mc E_t=e^{t \mc L}$, with $\mc L$ the generator in Lindblad form. Motivated by this observation, Refs. \cite{wolf2008dividing,wolf2008assessing} introduced a definition of Markovianity as a property of a \emph{single} quantum channel, rather than of a family of channels (i.e., of the evolution or process itself). Specifically, a channel $\mathcal E$ is called Markovian if there exists a generator $\mathcal L$ in Lindblad form such that $\mathcal E=e^\mathcal L$. For clarity, we refer to this notion of Markovianity as \emph{channel semigroup Markovianity} (CSM). Given an arbitrary channel $\mathcal E$ one can test whether it is CSM by computing $\log\mathcal E$ and verifying that it can be brought to Lindblad form~\cite{wolf2008assessing}. Note that the parameterization in Eq.~\eqref{eq:PL} describes a Pauli channel that is CSM by construction.

\begin{figure}[t]
 \includegraphics[width=\columnwidth]{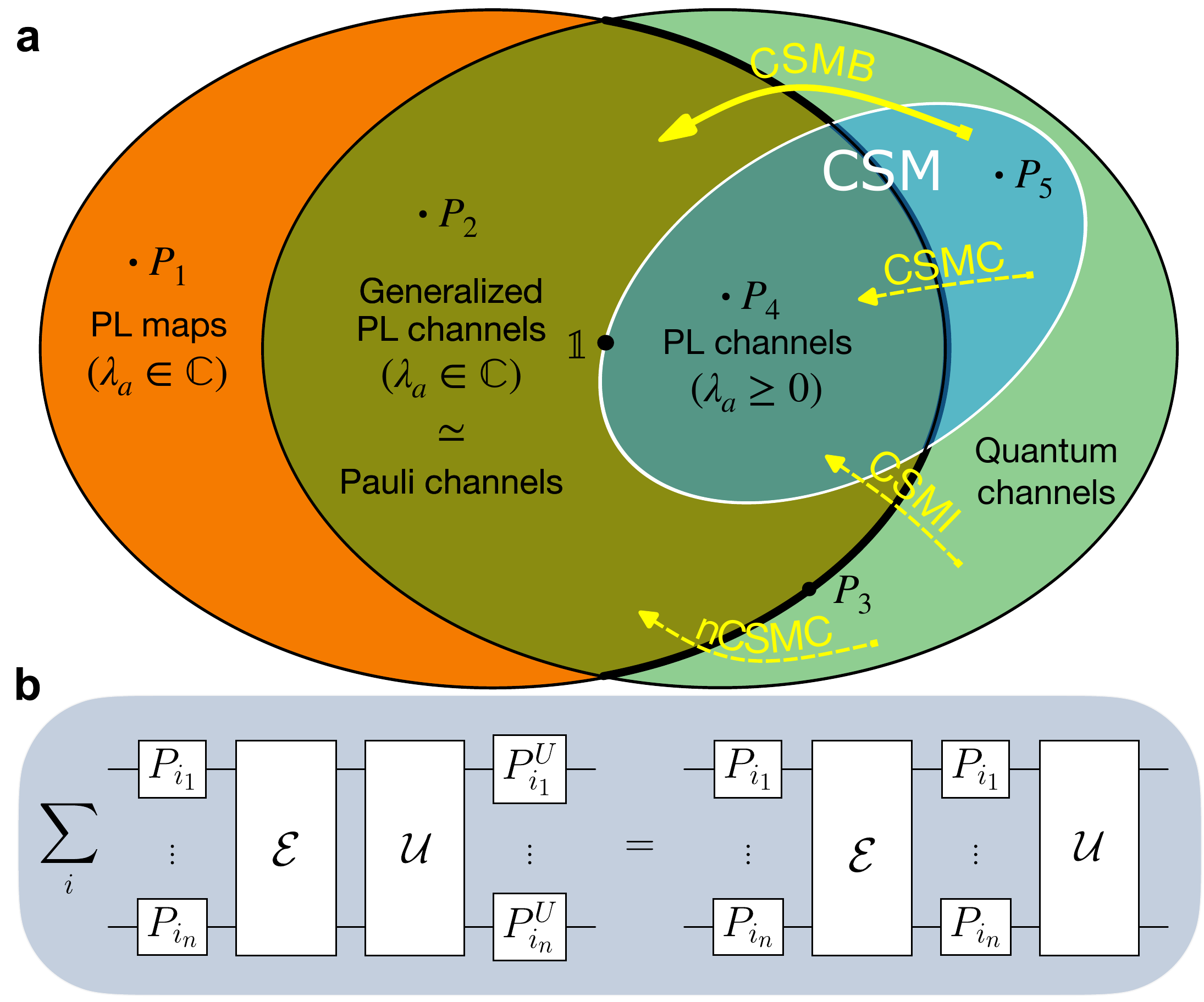}
 \caption{\label{fig:cartoon} \textbf{(a)} The relations between the maps in this paper. The Pauli-Lindblad (PL) maps (orange, $P_1$) are those maps of the form of \eq{PL} with $\la_a\in \mathbb C$. The generalized PL channels are defined as those PL maps that are also quantum channels (olive green, $P_2$). This set is essentially equal to the set of Pauli channels (the latter is the closure of the former, with limit points depicted as the boundary, $P_3$). The channel semigroup Markovian channels (teal, CSM) contain channels that are Pauli channels ($P_4$) and non-Pauli channels ($P_5$). The Pauli-Lindblad channels are essentially the CSM channels that are Pauli channels ($P_4$), forming a strict subset of the generalized PL channels. All possible nontrivial effects of Pauli twirling a quantum channel are indicated by the yellow arrows. Pauli twirling can break (CSMB), conserve (CSMC), or instate (CSMI) channel semigroup Markovianity.  It can also conserve the non-CSM property of a channel (nCSMC).
 \textbf{(b)} The noise channel of a noisy Clifford gate $\mc U\circ \mc E$, with $\mc U(\rho)=U\rho\,U^\dagger$ the noiseless Clifford, can be Pauli twirled by inserting a random Pauli word $P_a$ before, and the conjugated Pauli word $P^U_a=U P_a U^\dagger$ after the noisy gate.}
 \end{figure}

\paragraph{\normalfont\small\textbf{Pauli channels.}} 
A general quantum channel $\mc E$ admits a Kraus representation of the form $\mc E(\rho)=\sum_i K_i^\nodagger \rho\, K_i^\dagger$, where the Kraus operators $K_i$ satisfy the completeness relation $\sum_i K_i^\dagger K_i^\nodagger=\id$ \cite{Nielsen_Chuang_2010}.
The $n$-qubit Pauli words $P_a\in\{\id,X,Y,Z\}^{\otimes n}$  form an orthogonal operator basis, satisfying $\mathrm{tr}(P_aP_b)=2^n\delta_{ab}$. Each Kraus operator can be expanded in this basis as $K_i=\sum_a\ka_{ia} P_a$, so that we can write any channel as $\mc E(\rho)=\sum_{ab}p_{ab}P_a\rho P_b$, with $p_{ab}=\sum_i \ka_{ia}^{\phantom *}\ka_{ib}^*$. We say that the matrix $p_{ab}$ is the Pauli-basis representation of the channel.
In particular, a channel is called a Pauli channel if its Pauli-basis representation is diagonal, that is, if
$\mc E(\rho)=\sum_{a} p_a P_a \rho P_a$,
where $p_a$ is a probability distribution. Therefore, a Pauli channel is fully characterized by the distribution $p_a$.

An alternative representation of a general channel $\mc E$ is given by its transfer matrix
\begin{equation}\label{eq:transfer}
T_{ab}:=\frac{1}{2^n}\tr[P_a\mc E(P_b)].
\end{equation}
The diagonal elements of $T$ are referred to as the Pauli eigenvalues $f_a$, 
\begin{equation}\label{eq:peig}
f_a:=\frac{1}{2^n}\tr[P_a\mc E(P_a)].
\end{equation}
As with the matrix $p_{ab}$, Pauli channels have diagonal transfer matrices, thus the eigenvalues $f_a$ also fully characterize any Pauli channel. The Pauli eigenvalues $f_a$ and the probabilities $p_a$ are related by the Walsh-Hadamard transform
\begin{equation}
    f_a=\sum_k(-1)^{\langle a,k \rangle}p_k \xrightleftharpoons[]{} p_a=\frac{1}{4^n}\sum_k (-1)^{\langle a,k \rangle}f_k,\label{eq:ftop}
\end{equation}
where the symplectic inner product $\langle a,c \rangle$ equals $0$ if $P_a$ and
$P_c$ commute, and $1$ if they anticommute \cite{flammia2020efficient}.

Any quantum channel $\mc E$ can be converted into a Pauli channel $\mc E^{\mrm P}$ by Pauli twirling, 
which is performed by randomly injecting a Pauli word before, and the same random Pauli word after the channel,
\begin{equation}\label{eq:twirling}
\mc E^{\mrm P}(\rho):=\frac{1}{4^n}\sum_{a,i} P_a  K_i  P_a\rho P_a K_i^\dagger P_a.
\end{equation}
Using again the Pauli basis expansion of the Kraus operators, one finds that $\mc E^{\mrm P}$ is indeed a Pauli channel with $p_a=\sum_k|\ka_{ka}|^2$ \cite{flammia2020efficient}. In other words, Pauli twirling removes the off-diagonal elements of the Pauli-basis representation $p_{ab}$. Equivalently, in the transfer-matrix picture, Pauli twirling eliminates the off-diagonal elements of $T_{ab}$, yielding $T^{\mrm P}=\diag(f)$. 

\section{\label{sec:level2}Results}
\subsection{Generalized Pauli-Lindblad channels}
We first present  mathematical results, summarized in \fig{cartoon}a, on the relation between Pauli channels and channel semigroup Markovianity. We define \emph{PL maps} as those maps of the form of \eq{PL} with $\la_a\in \mathbb C$. It is straightforward to show that there exist PL maps that are not quantum channels. We therefore define the \emph{generalized PL channels} as those PL maps that are also quantum channels. As we show below, the set of generalized PL channels is essentially equal to the set of Pauli channels.

In terms of the PL parameters $\lambda_a$, the Pauli eigenvalues are given by
\begin{equation}\label{eq:ltof}
f_a=\exp\left(-2 \sum_{k}\la_k\langle a,k\rangle\right),
\end{equation}
as was found in Ref.~\cite{berg2023probabilistic}. We find that this expression can be inverted
\begin{equation}\label{eq:ftol}
\la_a=\frac{1}{4^n}\sum_{k\neq 0}(-1)^{\sp{a}{k}}\ln(f_k),
\end{equation}
with the definition $P_0=\id$. 
From Eq.~\eqref{eq:ftol}, it is clear that the PL parameters $\la_a$ are ill-defined for Pauli channels for which there is an $a$ such that $f_a=0$. Nevertheless, any such channel can be approached arbitrarily closely by a PL channel. Thus, we have the following result.

\begin{proposition}
The set of Pauli channels is the closure of the set of generalized PL channels.
\end{proposition}

We remark that for the identity channel, the Pauli eigenvalues are $f=(1,\ldots,1)$. As a result, Pauli channels with at least one vanishing Pauli eigenvalue $f_a=0$ represent strong-noise quantum channels and are not relevant for the characterization of the noise of high-fidelity quantum gates. Moreover, the eigenvalues satisfy $f_a\in[-1,1]$. Therefore, it is even possible for the $\la_a$ to be complex, which, again, occurs only for strong-noise quantum channels. Consequently, Eq.~\eqref{eq:ftol} also shows that the PL parameters $\la_a$ are not unique for these strong-noise channels, due to the possibility of choosing different branch cuts for the logarithm. This redundancy can be removed by consistently taking the principal branch of the logarithm, thus producing a unique representation in terms of $\lambda_a$ for each Pauli channel. Equations \eqref{eq:ltof} and \eqref{eq:ftol} then establish a bijection between the PL parameters $\la_a$ and Pauli eigenvalues $f_a$ ($f_a\neq 0$).

Given the bijection between the PL parameters and the Pauli eigenvalues, we have the following proposition, that we state due to its importance despite its simplicity. The proposition remains valid for strong-noise channels with possible negative Pauli eigenvalues, as long as $\mc P$ has nonzero Pauli eigenvalues; otherwise the PL parameters are ill-defined.   
\begin{proposition}
A Pauli channel $\mc P$ has real and nonnegative PL parameters if and only if $\mc P$ is channel semigroup Markovian (CSM). 
\end{proposition}
\begin{proof}
If a Pauli channel $\mc P$ has nonnegative PL parameters $\la$, then  $\mc P=e^{\mc L}$, where $\mc L$ is given in \eq{PL}, trivially in Lindblad form. Therefore, $\mc P$ is CSM. For the converse, assume that a Pauli channel $\mc P=e^\mc{L}$ is CSM with $\mc L$ as in \eq{PL}. Suppose that $\mc P$ has a PL parameter $\la_a$ that is not real and nonnegative. Then it cannot be brought to Lindblad form because $\la$ is unique and $\mc L$ is already in diagonal form with negative or complex eigenvalues $\la$, which contradicts the assumption that $\mc P$ is CSM.    
\end{proof}

Using Eqs.~\eqref{eq:ftop} and \eqref{eq:ftol} it is now possible to determine whether an $n$-qubit Pauli channel characterized by either its probability distribution $p_a$ or by its eigenvalues $f_a$ is CSM by checking the signs of the parameters $\lambda_a$. For example, for a single-qubit Pauli channel, the criterion according to Eq.~\eqref{eq:ftol} is $f_j\geq f_kf_l$ for all permutations $(j,k,l)$ of $(x,y,z)$. This condition for the single-qubit Pauli channel case was also reported in Ref.~\cite{Davalos2019}; our result generalizes this to the arbitrary $n$-qubit case.

\subsection{Gate frame}
Consider the noisy implementation $\Phi$ of a unitary quantum gate $\mc U$, with $\mc U(\rho)=U\rho\, U^\dagger$, where $U$ is the desired unitary. It is always possible to write 
\begin{equation}
\Phi=\mc U \circ \mc E, 
\end{equation}
which defines the error channel $\mc E$ of the noisy gate. If $U$ is a Clifford gate, so that $P^U_a U=UP_a$ for some Pauli word $P_{a}^U$, the noise channel $\mc E$ can be Pauli-twirled experimentally by picking a Pauli word $P_a$ uniformly at random and inserting the Pauli channel $\mc P_a(\rho)=P_a\rho P_a$ before and the channel $\mc P_a^U(\rho )=P_a^U
\rho P_a^U$ after the noisy gate \cite{Wallman2016rc},
\begin{equation}
\frac{1}{4^n}\sum_a\mc P_a^U \circ \Phi \circ \mc P_a=\mc U\circ \mc E^{\mrm P}, 
\end{equation}
as illustrated in Fig.~\ref{fig:cartoon}b. We wish to characterize the Pauli-twirled noise $\mc E^{\mrm P}$ of a quantum gate as a (generalized) PL channel.

Consider a standard gate-based quantum computer, where a quantum gate, taking time $t_g$, is described by a process that is Markovian by divisibility, i.e., where time evolution is described by~\eqref{eq:generator}. To perform a nontrivial quantum gate, the Hamiltonian $H(t)$ can be time-dependent. Moving to the frame of the gate, defined by $\tilde \rho(t)=U^\dagger(t) \rho(t) U(t)$, with $\rho(t)$ the state in the original frame, and $\dd U(t)/\dd t=-iH(t)U(t)$, we have
\begin{equation}
\Phi_t=\mc U_t \circ \mc E_t,
\end{equation}
with $\mathcal U_t(\rho)=U(t)\rho\, U^\dagger(t)$ the unitary evolution up to time $t$, and $\mc E_t$ the noise up to time $t$. The noisy quantum gate, with possible coherent and incoherent errors, is described by the channel $\Phi_{t_g}=\mc U_{t_g}\circ \mc E_{t_g}\equiv\Phi$. The noise channel $\mc E_t$ is generated by $\tilde{\mc L}_t=-i[\tilde H'(t),\,\cdot\,]+\tilde{\mc D}_t$. Here, $\tilde H'(t)$ describes any coherent control errors in the gate frame, and  $\tilde {\mc D}_t$ is the gate-frame dissipator, which is obtained by transforming the Kossakowski matrix,
\begin{equation}
\Gamma_{ab}(t)\rightarrow \tilde \Gamma_{ab}(t)=\sum_{c,d}T_{ac}^{U}(t)\Gamma_{cd}(t)T_{bd}^{U}(t),
\end{equation}
where $T_{ab}^{U}(t)$ is the transfer matrix of $\mathcal U_t$.

The Kossakowski matrix in the gate frame is thus generally time-dependent, even if it is time-independent in the original frame. Nevertheless, it remains positive semidefinite, so that $\tilde{\mc L}_t$ generates a family that is Markovian by divisibility. Because of the time dependence, this does not generally imply that $\mc E=\mc E_{t_g}$ (which is generated by $\tilde{\mc L_t}$) is CSM, nor does it imply that $\mc E^{\mrm P}$ is CSM. In the following, we give examples showing that even if $\mc E$ is CSM, Pauli twirling may break this property (CSMB). We additionally show that Pauli twirling may conserve (CSMC) and instate (CSMI) the CSM property of $\mc E$, and may even conserve non-CSM (nCSMC).  

\subsection{Hadamard dephasing}

Consider a qubit channel $\mc E_t$ generated starting from time $t=0$ by a purely dissipative generator [$H(t)=0$], with a single jump operator being the Hadamard matrix,
\begin{equation}
L_\varphi=(X+Z)/\sqrt{2},
\end{equation}
and associated constant dephasing rate $\gamma_\varphi>0$. (In this example, the gate frame equals the original frame since $H(t)=0$). Since the jump operator $L_\varphi$ is time-independent, the channel $\mc E_t=e^{t \mc L}=e^{t \mc D}$ is manifestly CSM for any $t\geq 0$. Computing the Pauli eigenvalues $f_a$ and using Eq.~\eqref{eq:ftol}, we find that the PL parameters of $\mc E_t^{\mrm P}$ are
\begin{equation}
\begin{aligned}
\la_x(t)&=\la_z(t)=\frac{\gamma_\varphi t}{2},\\
\la_y(t)&=-\frac{1}{2}\ln[\cosh(\gamma_\varphi t)].
\end{aligned}
\end{equation}
We note that $\lambda(t)$ should be interpreted as specifying the value of the PL parameter for a channel of duration $t$.
For any $\gamma_\varphi t\geq 0$, the PL parameter $\la_y(t)$ is negative. Therefore, Pauli twirling may break channel semigroup Markovianity (CSMB). Given a PL qubit channel with two equal PL parameters with value $\ell$, the third must be at least $-\frac{1}{2}\ln[\cosh(2\ell)]$. This is shown by setting two PL parameters to $\ell$ (e.g. $\la_x=\la_z=\ell$), and  using Eqs.~\eqref{eq:ftop} and \eqref{eq:ltof} to find the $\la_y$ for which $p_a$ remains a probability distribution. Hadamard dephasing saturates this bound.

\subsection{Hadamard dephasing and relaxation}
We now extend Hadamard dephasing with a relaxation process. Hadamard dephasing can be viewed as a rotated version of standard dephasing because $L_\varphi=R_y(\pi/4)ZR_y(-\pi/4)$, where $R_y(\vartheta)$ denotes a rotation along the $y$-axis by angle $\vartheta$. Analogously, we  define a rotated relaxation process with the jump operator $L=R_y(\pi/4)\sigma_-R_y(-\pi/4)$, where $\sigma_-$ is the qubit lowering operator, and we denote the corresponding relaxation rate by $\gamma$. Considering a purely dissipative generator with jump operators $L_\varphi$ and $L$, we obtain the channel $\mc E_t$. Computing the Pauli eigenvalues and using \eq{ftol}, we find that the Pauli-twirled channel $(\mc E_t)^{\mrm P}$ is parameterized by 
\begin{align}
\la_x(t)&=\la_z(t)=\frac{\gamma_\varphi t}{2} +\frac{\gamma t}{8},\\
\la_y(t)&=\frac{\gamma t}{4}-\frac{1}{2}\ln\left[\cosh\left(\gamma_\varphi t-\frac{\gamma t}{4}\right)\right].\label{eq:had_lay}
\end{align}
This shows that relaxation competes with dephasing in making $\lambda_y$ negative, resulting in parameter regimes for which $\la_y\geq 0$ and parameter regimes for which $\la_y< 0$. This illustrates that Pauli twirling can also conserve the CSM property of a (non-Pauli) quantum channel (CSMC).

We expand the right-hand side of \eq{had_lay} around $\gamma_\varphi=\gamma=0$ to second order in $\gamma_\varphi$ and $\gamma$. Setting the result to zero, solving for $\gamma t$, and re-expanding to second order in $\gamma_\varphi t$, we find that the condition for $\lambda_y(t)<0$ is
\begin{equation}
\gamma t \lesssim (\gamma_\varphi t)^2.
\end{equation}
The relaxation and dephasing rates $\gamma$ and $\gamma_\varphi$ are related to the relaxation and (Ramsey) dephasing times $T_1$ and $T_2^*$ by $\gamma t=t/T_1$, $\gamma_\varphi t=t/T_2^*-t/(2T_1)$.
Similar steps after first making these substitutions in  \eq{had_lay} lead to
\begin{equation}
    \frac{t}{T_1}\lesssim\left(\frac{t}{T_2^*}\right)^2.
\end{equation}
Interestingly, if the noise is insufficiently biased towards dephasing, this can be compensated by longer idling times.

Although we have thus far introduced Hadamard dephasing and relaxation as a toy model to show negative PL parameters, we stress that it can be realized experimentally under realistic conditions nevertheless. Considering a qubit that undergoes Markovian dephasing and relaxation, the tilted relaxation and dephasing processes are implemented by applying a $\pi/4$ pulse to the qubit, allowing it to idle for a duration $t$, and then rotating it back with a $-\pi/4$ (or $7\pi/4$) pulse. 
Given any nonzero relaxation rate, one can explore both sides of the parameter regime by sweeping the idling time $t$ and observing the sign change of $\lambda_y$. This gives an exciting opportunity to experimentally observe Pauli-twirling induced non-Markovianity. 

Though experimentally feasible, Hadamard dephasing and relaxation is designed as an illustrative example. Since the unitary $\mc U$ involved is the identity, it does not represent a relevant quantum gate. We therefore proceed by considering a nontrivial quantum gate.

\subsection{The $\sqrt{X}$ gate}
Consider a qubit subjected to a circularly polarized resonant pulse with drive strength $A$ in the presence of Markovian relaxation and dephasing. In the frame rotating at the drive frequency, the qubit dynamics is governed by~\eq{generator}, where the Hamiltonian $H=AX/2$ is time-independent. The same rotating-frame Hamiltonian is obtained if the pulse is linearly polarized under the rotating wave approximation \cite{Scully_Zubairy_1997}. The jump operators consist of $Z$, with associated time-independent dephasing rate $\gamma_\varphi \ge 0$, and the qubit lowering operator $\sigma_-=\ket{0}\!\bra{1}$, with associated time-independent relaxation rate $\gamma \ge 0$. The closed evolution, i.e., when $\gamma=\gamma_\varphi=0$, describes a rotation along the $x$-axis, $R_x(\vartheta)$, with $\vartheta=At$ being the rotation angle. The $\sqrt{X}$ gate is implemented at $t_g=\pi/(2A)$. 

In the gate frame, the time evolution is generated by $\tilde{\mc L}_t={\tilde {\mc D_t}}$, where the jump operators are  $\tilde \si_-(t)=e^{iH_\mathrm{}t}\si_-e^{-iH_\mathrm{}t}$ and $\tilde Z(t)=e^{iH_\mathrm{}t}Ze^{-iH_\mathrm{}t}$, and the decoherence rates are unchanged.
The generated family of noise channels is therefore manifestly Markovian by divisibility. However, due to the time-dependent jump operators, $\mc E_{t_g}$ is generally not CSM. Perhaps surprisingly, using the methods from Ref.~\cite{wolf2008assessing}, we find numerically that for certain parameter regions of $\gamma_\varphi t_g, \gamma t_g$, including $0\leq \gamma_\varphi t_g\leq2,0\leq \gamma t_g\leq 2$, the noise channel $\mc E_{t_g}$ is CSM nonetheless (see \fig{contours}).

\begin{figure}[t]
 \includegraphics[width=\columnwidth]{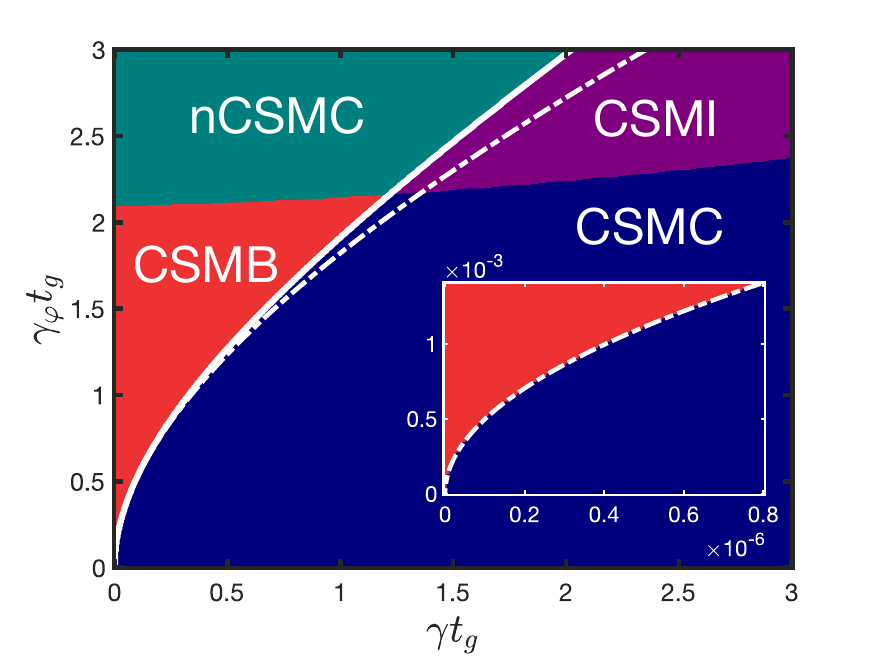}
 \caption{\label{fig:contours} Transition diagram of channel semigroup Markovianity for the error channel of a $\sqrt{X}$ gate due to Pauli twirling. The two-dimensional parameter space $(\gamma t_g,\gamma_\varphi t_g)$ is divided into four regions corresponding to the effects of Pauli twirling on channel semigroup Markovianity. 
 In the region labeled as CSMB (red) the initial CSM property is broken by Pauli twirling and the resulting channel can only be described by a generalized PL channel with a negative PL parameter. In region CSMC (dark blue) CSM is conserved and the resulting channel has only nonnegative PL parameters. At extreme dephasing rates ($\gamma_\varphi t_g\gtrsim2$) the error channel is non-CSM; Pauli twirling instates Markovianity in region CSMI (purple) and conserves non-Markovianity in region nCSMC (teal). The solid white line is the boundary separating the regions with different channel semigroup Markovianity after Pauli twirling, the dashed white line depicts the implicit curve $\lambda_x(t_g)=0$ [Eq.~\eqref{eq:sqrtx_lax}], the second-order approximation of this boundary.
 Most importantly, region CSMB contains  experimentally relevant scales, shown in the inset, with excellent agreement between the numerical and analytical second-order boundary.}  
 \end{figure}

We calculate the Pauli eigenvalues $f_a$ and the corresponding PL parameters $\la_a$ for the Pauli-twirled error channel $\mc E_{t_g}^{\mrm P}$ utilizing transfer matrix representations and a cumulant expansion (see Methods). To second order [i.e., also ignoring terms proportional to $\gamma t_g(\gamma_\varphi t_g)^2$, $(\gamma t_g)^2\gamma_\varphi t_g$, and $(\gamma t_g)^2(\gamma_\varphi t_g)^2$], we find
\begin{equation}
\begin{aligned}
\la_x(t_g)&\approx\frac{\gamma t_g}{4}-\frac{\sin^4\vartheta}{4\vartheta^2}\left(\gamma_\varphi t_g-\frac{\gamma t_g}{4}\right)^2,\\
\la_{y,z}(t_g)&\approx\frac{\gamma t_g}{8}\left(1\pm\frac{\sin2\vartheta}{2\vartheta}\right)+\frac{\gamma_\varphi t_g}{2}\left(1\mp\frac{\sin2\vartheta}{2\vartheta}\right).
\end{aligned}\label{eq:lambdagate}
\end{equation}
Using the scheme in \fig{cartoon}b to twirl the noise channel of a noisy gate, the above PL parameters can only be observed experimentally at $\vartheta=m \pi/2$, because only then the resulting gate is a Clifford gate. For the $\sqrt X$ gate specifically, $\vartheta=\pi/2$, in which case
\begin{align}\label{eq:sqrtx_lax}
    \la_x(t_g)&\approx\frac{\gamma t_g}{4}-\frac{1}{\pi^2}\left(\gamma_\varphi t_g-\frac{\gamma t_g}{4}\right)^2,\\
    \la_y(t_g)&=\la_z(t_g)=\frac{\gamma t_g}{8}+\frac{\gamma_\varphi t_g}{2}.
\end{align}
The expressions for $\la_y$ and $\la_z$ are exact; at $\vartheta=\pi/2$ the higher order terms drop. 

Based on \eq{sqrtx_lax},  computations as in the previous section result in $\la_x<0$ if
\begin{equation}
    \gamma t_g\lesssim\left(\frac{2}{\pi}\right)^2(\gamma_\varphi t_g)^2,\label{eq:criterion}
\end{equation}
or in terms of dephasing and relaxation times,
\begin{equation}
    \frac{t_g}{T_1}\lesssim\left(\frac{2}{\pi}\right)^2\left(\frac{t_g}{T_2^*}\right)^2.
\end{equation}
This is as for Hadamard relaxation and dephasing, but, at equal dehasing times, $\gamma t_g$ needs to be smaller by a factor $(2/\pi)^2\approx 0.4$. 

Thus, negative PL parameters may be observed for the noise channel of the  vanilla implementation of the $\sqrt{X}$ gate, provided that the noise is sufficiently biased towards dephasing. Platforms whose noise is strongly biased towards dehphasing include semiconductor spin-qubit, NV center, ion trap, and neutral atom platforms. As a concrete example, consider spin qubits in silicon, where typically $t_g\sim 100\,\mrm{ns}$, $T_1\gg1\,\mrm{ms}$, and $T_2^*\sim 1\text{-}10\,\mu\mrm{s}$~\cite{Silicon2022}. This gives $\frac{t_g}{T_1}\ll10^{-4}$ and $\left(\frac{2}{\pi}\right)^2\left(\frac{t_g}{T_2^*}\right)^2\sim 10^{-4}\text{-}10^{-2}$, showing that the condition on the noise bias can be met. If on any given platform the noise bias is insufficient, a possible remedy is pulse stretching, i.e., lowering the drive strength $A$ and thereby increasing $t_g$.

\section{\label{level3}Discussion}

\subsection{\label{sec:level3.1}
Impact on quantum error mitigation}

The goal of a quantum computation is often to determine the expectation value
$\langle \mathcal O \rangle$ of an observable $\mathcal O$ with respect to a
(pure) quantum state. A quantum computer applies a quantum circuit to prepare a noisy version of this state and then measures $\mathcal O$. Quantum error mitigation (QEM) aims to reduce the effects of noise through repeated circuit execution and classical coprocessing \cite{QEMreview2023}. Two prominent QEM techniques are probabilistic error cancellation (PEC) and zero-noise extrapolation (ZNE)~\cite{temme2017error,QEMreview2023}.

Following Refs.~\cite{vdBevidence2023,berg2023probabilistic}, when PL models are used in QEM, both PEC and ZNE proceed as follows. After each circuit layer, the map $(\mathcal E^{\mathrm P})^\beta$ is applied by injecting Pauli words into the circuit, where $\mathcal E^{\mathrm P}$ is the Pauli-twirled error channel of the layer. The map $(\mathcal E^{\mathrm P})^\beta$ has PL parameters $\lambda_a' = \beta \lambda_a$, with $\lambda_a$ the PL parameters of $\mathcal E^{\mathrm P}$. In Refs.~\cite{vdBevidence2023,berg2023probabilistic}, these parameters are assumed to be nonnegative. The injected Pauli words are sampled by writing $(\mathcal E^{\mathrm P})^\beta$ as
\begin{equation}
(\mathcal E^{\mathrm P})^\beta
= \bigcirc_a ( w_a \,\cdot + (1-w_a)\, P_a \cdot P_a ),
\end{equation}
where $w_a = (1 + e^{-2\lambda_a'})/2$. For each $a$, the Pauli operator
$P_a$ is chosen with probability $1 - w_a$. The product of all sampled $P_a$ yields a single Pauli word, which is then injected into the circuit. This procedure is repeated for each circuit layer and for each circuit execution.

In ZNE, the noise is amplified by choosing $\beta > 0$. The expectation values $\langle \mathcal O \rangle$ are measured for several values of $\beta$.
These data points are then extrapolated to the noiseless limit $\beta = -1$, where $(\mathcal E^{\mathrm P})^\beta$ exactly cancels $\mathcal E^{\mathrm P}$.

In PEC, $\beta = -1$ is chosen from the outset. This implies
$\lambda_a' < 0$ for all $a$, so that $1 - w_a$ can no longer be
interpreted as a probability. This issue can be remedied by
quasi-probabilistic sampling, as explained in Refs.~\cite{temme2017error,berg2023probabilistic}. In the terminology of this paper, $(\mathcal E^{\mathrm P})^{-1}$ remains a PL map.

The discussion of ZNE and PEC above assumes nonnegative PL parameters for $\mc E^{\mrm P}$. However, as we have shown in this paper, excluding negative PL parameters can lead to a fundamental mischaracterization $\check{\mathcal E}^{\mathrm P}$ of the true error channel $\mathcal E^{\mathrm P}$. In this case, an incorrect map $(\check{\mathcal E}^{\mathrm P})^\beta$ is applied. This results in an incorrect estimate of $\langle \mathcal O \rangle$ at $\beta = -1$ for both ZNE and PEC.

Nevertheless, ZNE and PEC based on PL models remain wholly possible when negative PL parameters are allowed. Quasi-probabilistic sampling can be straightforwardly extended to accommodate maps for which $0 \le w_a \le 1$ for some $a$, while $w_a > 1$ for others. This extension applies to both ZNE and PEC, since both types of coefficients may now occur in these QEM protocols.

\subsection{Conclusion}

We have shown that negative Pauli–Lindblad (PL) parameters are necessary for the faithful description of non-Markovian Pauli channels, where non-Markovianity is understood as the violation of the channel semigroup property. Importantly, such non-Markovian Pauli channels can arise even when the underlying physical noise model is Markovian (by divisibility), and even when the resulting noise channel is (channel-semigroup) Markovian, because Pauli twirling can break the (channel semigroup) Markovianity property. This effect occurs, for example, in the standard implementation of the $\sqrt{X}$ gate on dephasing-dominated platforms, and thus provides a concrete opportunity to experimentally observe Pauli-twirling-induced non-Markovianity in a standard implementation of a standard quantum gate. Allowing for negative PL parameters is essential in quantum error mitigation, as excluding this possibility may lead to incorrect noise characterization, and thereby incorrect error-mitigated results. 

Perhaps in hindsight, it is not so surprising that Pauli-twirling can break Markovianity (by the channel semigroup property); to Pauli-twirl a channel, a Pauli word $P_a$ is inserted before the channel for an $a$ drawn uniformly at random. The channel acts, and in the meantime, an auxiliary classical degree of freedom has to ``memorize'' $a$, so that $P_a$ can be injected after the channel has acted. This auxiliary degree of freedom lies outside the system. It is hence part of the environment, which is consequently not memoryless.

Two factors likely contribute to the fact that negative PL parameters have not been observed experimentally thus far. First, the possibility of negative PL parameters has been excluded a priori by assumption through a constraint in least-squares fitting procedures applied to experimental data because the negative PL parameters were thought to be unphysical \cite{vdBevidence2023}. Second, PL parameters have not yet been experimentally characterized in dephasing-dominated platforms. As such, this work also serves as practical guidance when characterizing error channels using PL models: when fitting PL parameters to experimental data, one must explicitly allow for the possibility of negative values, even if the physical noise is assumed to be Markovian (by divisibility). 

The two illustrative examples that we have considered, namely, a tilted dephasing and relaxation process, and the noisy $\sqrt{X}$ gate, are both single-qubit processes. We emphasize, however, that our theory applies equally to multi-qubit gates and that negative PL parameters may similarly arise there, providing a natural direction for future investigation.

\section{Methods}
\paragraph{\normalfont\small\textbf{Perturbation theory of PL parameters.}}In the rotating frame, the noisy implementation of the $R_x(\vartheta)$ gate is $\Phi=e^{t_g\mathcal L}$, where $\mathcal L$ is given by Eqs.~\eqref{eq:generator} and \eqref{eq:dissipator}, with a time-independent Hamiltonian $H=AX/2$ and time-independent Kossakowski matrix, which, using the Pauli matrices as the operator basis, reads
\begin{equation}
    \Gamma_{}=\frac{1}{4}\begin{pmatrix}
        \gamma&i\gamma&0\\
        -i\gamma&\gamma&0\\
        0&0&4\gamma_\varphi
    \end{pmatrix}.
\end{equation}

The transfer matrix of $\Phi$ is given by the exponential $T_\Phi=e^{t_gT_{\mc L}}$, where $T_{\mc L}$ denotes the transfer matrix of $\mathcal L$. Decomposing $T_{\mc L}$ into its  Hamiltonian and dissipative components, $T_{\mc L}=T_\mc H+T_{\mc D}$, we may write $T_\Phi=e^{t_g(T_\mc H+T_\mc D)}=e^{t_gT_{\mc H}}e^\mc C$, where $\mc C$ arises from a cumulant-like expansion~\cite{Feynman1951,Wilcox1967,Emerson2005}. Retaining terms up to second order in the parameters $\gamma t_g$ and $\gamma_\varphi t_g$, we obtain
\begin{align}
    \mc C^{(1)}=\int_0^{t_g}\mathrm ds\ e^{-sT_{\mc H}}T_{\mc D}e^{sT_{\mc H}}\equiv\int_0^{t_g}\mathrm ds\ T_{\mc D}(s),\\
    \mc C^{(2)}=\frac{1}{2}\int_0^{t_g}\mathrm ds_1\ \int_0^{s_1}\mathrm ds_2\ \left[T_{\mc D}(s_1),T_{\mc D}(s_2)\right].
\end{align}
The resulting error channel of the noisy gate is therefore described by $T_\mc E=e^{-t_gT_\mathcal H}T_\Phi=e^\mc C\approx e^{\mc C^{(1)}+\mc C^{(2)}}$.

Expanding the exponential to second order in the parameters $\gamma t_g$ and $\gamma_\varphi t_g$, we obtain the following perturbative approximation for the transfer matrix of the error channel,
\begin{align}
    T_\varepsilon\approx\mathds 1+\mathcal C^{(1)}+\frac{\left(\mathcal C^{(1)}\right)^2}{2}+\mathcal C^{(2)}.
\end{align}
The diagonal entries of this matrix correspond to the series expansion of the Pauli eigenvalues to second order, $f_k\approx1+f_k^{(1)}+f_k^{(2)}$. Finally, Eq.~\eqref{eq:ftol} and the second order expansion of the logarithm, $\ln f_k\approx f_k^{(1)}+f_k^{(2)}-\left(f_k^{(1)}\right)^2/2$, lead to the perturbative expressions for the PL parameters $\lambda_k$ in Eq.~\eqref{eq:lambdagate}.

\subsection{Acknowledgments}
The authors acknowledge support from 
the German Federal Ministry of Research, Technology and Space
(BMFTR) under the QSolid project, Grant No.~13N16167, 
the Ministry of Economic Affairs, Labour and
Tourism Baden-Württemberg within the Competence Center
Quantum Computing (KQCBW24),
Deutsche Forschungsgemeinschaft (DFG) under Project No.~425217212 -- SFB 1432,
and the Munich Quantum Valley (K-8), which is supported by the Bavarian state government with funds from the Hightech Agenda Bayern Plus.
\bibliography{bib}
\end{document}